\documentclass[aps,pra,floatfix,amsmath,superscriptaddress,twocolumn,nofootinbib]{revtex4-2}
\usepackage{amssymb}
\usepackage{graphicx}
\usepackage{graphics}
\usepackage{amsmath}
\usepackage{amsthm}
\usepackage{color}

\newcommand{\bea}{\begin{eqnarray}}
\newcommand{\eea}{\end{eqnarray}}

\newcommand{\s}{S(\lambda_1,\lambda_2,\lambda_3,\lambda_4)}
\def\bi{\begin{itemize}}
\def\ei{\end{itemize}}
\def\bc{\begin{center}}
\def\ec{\end{center}}

\def\C{\hbox{$\mit I$\kern-.7em$\mit C$}}
\def\R{\hbox{$\mit I$\kern-.6em$\mit R$}}

\newcommand{\one}{\mbox{$1 \hspace{-1.0mm}  {\bf l}$}}
\def\tr{\mathrm{tr}}

\def\supp{\textrm{supp\,}}

\def\eig{\textrm{eig\,}}

\newtheorem{theorem}{Theorem}

\newtheorem{lemma}[theorem]{Lemma}

\begin{document}

\title{Nonexistence of maximally entangled mixed states for a fixed spectrum}
\author{Gonzalo Camacho}\email{gonzalo.camacho@dlr.de}
\affiliation{Department High-Performance Computing, Institute of Software Technology, German Aerospace Center (DLR), 51147 Cologne, Germany}
\author{Julio I. de Vicente}\email{jdvicent@math.uc3m.es}
\affiliation{Departamento de Matem\'aticas, Universidad Carlos III de
Madrid, E-28911, Legan\'es (Madrid), Spain}
\affiliation{Instituto de Ciencias Matem\'aticas (ICMAT), E-28049 Madrid, Spain}

\begin{abstract}
The existence of a maximally entangled pure state is a cornerstone result of entanglement theory that has paramount consequences in quantum information theory. A natural generalization of this property is to consider whether a notion of maximal entanglement is possible among all states with the same spectrum (where the aforementioned case of pure states corresponds to the particular choice in which the spectrum is a delta distribution, i.e., rank-1 states). Despite positive evidence in the past that such a notion might exist at least in the case of two-qubit states, it was recently shown in [Phys. Rev. Lett. \textbf{133}, 050202 (2024)] that the answer to the above question is negative. This reference proved this for particular choices of the spectrum in the case of rank-2 two-qubit density matrices. While this settles the problem in general, it still leaves open whether there are other choices of the spectrum outside the case of pure states where a maximally entangled state for a fixed spectrum might exist. In this work we extend this impossibility result to all rank-2 and rank-3 two-qubit states as well as for a large class of eigenvalue distributions in the case where the rank equals four. 
\end{abstract}

\maketitle

\section{Introduction}

The study of entanglement plays a central role in the foundations of quantum mechanics. This interest has been boosted in the last decades due to its applications in quantum information theory, where it is regarded as a resource that provides an advantage in order to implement certain tasks. A crucial question from this point of view is to identify the entangled state that maximizes the performance of a given task. In particular, the analysis becomes simpler if there exists a state that is optimal irrespectively of the particular task, giving thus rise to a universal notion of maximal entanglement.

These questions can be addressed in the context of entanglement theory (see e.g.\ the review articles \cite{review1,review2,review3}), which is formulated as a quantum resource theory \cite{resource}. In quantum information theory entangled states are used by spatially separated parties and, hence, bound to protocols pertaining to the class of local operation and classical communication (LOCC). This immediately induces an operational order in the set of entangled states %. If there exists an LOCC protocol transforming an entangled state $\rho$ into another entangled state $\sigma$, then the latter cannot be more entangled than the former. This is because $\sigma$ cannot provide a bigger advantage than $\rho$ for any task to be implemented by LOCC protocols, as any protocol relying on $\sigma$ can be also carried out by sharing $\rho$ by initially transforming $\rho$ into $\sigma$. This moreover 
and it makes it possible to quantify entanglement. An entanglement measure is given by any function from the set of states to the non-negative real numbers that preserves the order induced by LOCC manipulation and maps all non-entangled (i.e., separable) states to zero.

Thus, the study of LOCC protocols is a pivotal part of entanglement theory and Nielsen's theorem \cite{nielsen}, which characterizes LOCC convertibility between bipartite pure states, is a central result. In particular, it implies that LOCC induces in general only a partial order in the set of entangled states. That is, there exists pairs of incomparable states under LOCC and there cannot be a unique entanglement measure. Instead, there are different entanglement measures, which can be provided with different operational meanings in relation to the particular kind of advantage one quantifies, and, in general, the notion of some state being more useful than another can only be made task-dependent. This notwithstanding, a salient feature of this theorem is that this notion can be made task-independent for pairs that are related by the partial order. In particular, it turns out that the two-qudit generalized Bell state,
\begin{equation}
|\phi^+_d\rangle=\frac{1}{\sqrt{d}}\sum_{i=1}^d|ii\rangle,
\end{equation}
can be transformed by LOCC into any other two-qudit state (pure or not). According to the above discussion, it then follows that $|\phi^+_d\rangle$ is the maximally entangled two-qudit state. The existence of a universal notion of maximal entanglement is a very pleasant feature of the theory of bipartite entanglement. Whether it is e.g.\ teleportation or entanglement-assisted state discrimination, the above state must be the most useful state for any task to be implemented by LOCC protocols once the underlying local dimension $d$ of our physical system is fixed. This also entails that $|\phi^+_d\rangle$ maximizes all entanglement measures among all two-qudit states and it provides a way to gauge the entanglement content of all other states.

Therefore, in an ideal scenario we would always choose to distribute the maximally entangled state, as it is universally optimal for any LOCC task. Suppose, however, that we have a device that only produces states in some set $S$, which is a strict subset of the set of all 2-qudit states. If $|\phi^+_d\rangle\notin S$, then it is not a priori clear what the most entangled state in $S$ is and it could be that our choice had to be conditioned on the task to be implemented, i.e., on a particular choice of entanglement measure to be optimized. Nevertheless, a universal (i.e., task-independent) notion of maximal entanglement in $S$ might be possible and the principles of entanglement theory clearly dictate how to formulate this. We will say, if it exists, that $\rho\in S$ is \textit{the maximally entangled state in the set} $S$ if for every state $\sigma\in S$, there exists an LOCC map $\Lambda$ such that $\Lambda(\rho)=\sigma$. One can easily envision situations where this problem might arise in practice, in particular because in this case one is bound to mixed states and perfect pure states such as $|\phi^+_d\rangle$ cannot be prepared. A particularly relevant instance of this problem occurs when one studies entanglement generation schemes by Hamiltonian (i.e., joint unitary) evolution on an input separable state such as those considered in \cite{unitary1,unitary2}. If this input state is pure, we would always choose to engineer the unitary evolution to lead to the maximally entangled state $|\phi^+_d\rangle$. However, in the presence of noise the input state will become mixed and, even if we have perfect control of the induced evolution, in this case the output could never be a pure state such as $|\phi^+_d\rangle$. In fact, under the above assumption, we can and only can prepare states that have the same spectrum as the initial separable input state.

To formalize this problem, we consider the simplest case of two-qubit states and we define $S(\lambda_1,\lambda_2,\lambda_3,\lambda_4)$ to be the set of all two-qubit density matrices with ordered eigenvalues given by $(\lambda_1,\lambda_2,\lambda_3,\lambda_4)$% (without loss of generality arranged in non-increasing order and satisfying that $\sum_i\lambda_i=1$ and $\lambda_i\geq0$ $\forall i$)
. The question we want to address is for any given spectrum whether there exists a maximally entangled state in $\s$, and, if so, what could this state be. Notice that, according to what we have been discussing, $|\phi^+_2\rangle$ is the maximally entangled state in $S(1,0,0,0)$, so this question is also a natural generalization of the above. Interestingly, Ref.~\cite{mems1} showed that for any choice of eigenvalue distribution the same state in $\s$ is the unique (up to local unitary equivalence) maximizer among all states in $\s$ of three different important entanglement measures: the entanglement of formation \cite{eof}, the relative entropy of entanglement \cite{ree} and the negativity \cite{negativity}. If a maximally entangled state in $\s$ existed, it has to maximize all entanglement measures among $\s$, so this was taken as positive evidence that the family of states found in \cite{mems1} could provide the maximally entangled states in $\s$ and they have been referred to in the literature as MEMSs (maximally entangled mixed states). However, against this evidence, it has been recently proven in \cite{nomems} that this is not the case and that a universal notion of maximal entanglement for a fixed spectrum does not always exist. Namely, it has been shown therein that there is no maximally entangled state in the set $S(\lambda,1-\lambda,0,0)$ whenever $\lambda\in(2/3,1)$. This settles the question in general and shows that there is unfortunately not always a universally optimal choice of output state in the scenario of entanglement generation by Hamiltonian evolution described above. Hence, this choice must be conditioned on a particular task in which one wants to maximize the entanglement-based advantage. Nevertheless, the result of \cite{nomems} still leaves open whether there are other choices of the spectrum for which a maximally entangled state in $\s$ might exist beyond the case of $S(1,0,0,0)$. This is the goal of the present work. 

A possible way to show that there is no maximally entangled state in $\s$ is to find an entanglement measure different to those considered in \cite{mems1} that has a different maximizer than the corresponding MEMS. However, the explicit computation of entanglement measures relies on hard optimization problems and this seems to be a daunting task. The approach followed in \cite{nomems} goes instead to the very definition of maximally entangled state in a set. Since entanglement measures cannot increase under LOCC, the result of \cite{mems1} implies that no isospectral local-unitary-inequivalent state can be transformed by LOCC into the corresponding MEMS. Thus, if we find a single instance of an isospectral local-unitary-inequivalent state that cannot be obtained either by LOCC from the corresponding MEMS, we can conclude that there is no maximally entangled state in $\s$ for that eigenvalue distribution. Nevertheless, the problem of deciding when an LOCC protocol exists transforming a given mixed state into a given mixed target state is also known to be highly non-trivial. In fact, there is no generalization of Nielsen's theorem to this case and only the particular instances of probabilistic \cite{prob} and approximate \cite{approx} transformations among pure states and transformations from pure states to ensembles \cite{ensemble} and mixed states \cite{mixed} have been characterized. Reference \cite{nomems} obtains its result by considering a well-known relaxation of LOCC to the superset of the so-called non-entangling (NE) maps \cite{harrownielsen,brandaoplenio1,brandaoplenio2,brandaodatta,patricia,beyondlocc,lamiregula1,brandaoplenio3,lamiregula2}. In addition to making it possible to find a relatively simple counterexample of a target state that cannot be obtained by NE maps from the MEMS when the spectrum is given by $(\lambda,1-\lambda,0,0)$ with $\lambda\in(2/3,1)$, one obtains as a by-product that a maximally entangled state for a fixed spectrum in this case cannot exist in any resource theory of entanglement. This is because, according to the formalism of quantum resource theories \cite{resource}, the maximal set of transformations into which one can relax LOCC in the single-copy regime is NE. Thus, while, for instance, the resource theory of pure multipartite entanglement has a cleaner picture by considering this relaxation \cite{patricia}, this is not the case for the problem at hand. 

In this paper, we put forward completely different techniques than those used in \cite{nomems} by considering a different but also well-known relaxation of LOCC: the class of separable (SEP) maps \cite{Rains1,Rains2,GG,CD,GW,Hebenstreit1,Hebenstreit2}. This allows us to find isospectral states that cannot be obtained by LOCC from the MEMS of \cite{mems1} for all eigenvalue distributions corresponding to the cases of rank equal to two and three. Thus, we conclude that there cannot exist a maximally entangled state in the set $S(\lambda_1,\lambda_2,\lambda_3,0)$ (if $\lambda_2\neq0$). Unfortunately, these techniques do not seem to extend easily to the full-rank case. For this reason, and also because of that pointed out above of whether this notion is possible in a less restrictive resource theory of entanglement, we also generalize the techniques used in \cite{nomems} based on NE transformations. This makes it possible to write the convertibility question as a linear program \cite{boyd} and perform a systematic analysis. While this does not allow us to exclude the existence of a maximally entangled state in $\s$ for all possible eigenvalue distributions, this reproduces to a large extent the results in the rank-deficient cases under SEP and discards a large class of spectra in the full-rank case even under this more permissive class of operations. Thus, our results suggest that a maximally entangled two-qubit state for a fixed spectrum never exists outside the pure-state case, and, if this is not the case, it can only happen in very particular cases. In addition to this, we believe that our results might be of independent interest for the general and relatively unexplored problem of discerning LOCC convertibility in the case of mixed states.

This article is structured as follows. We set our notation and provide basic definitions in Sec.\ II. In Sec.\ III we consider SEP transformations among mixed states and prove the non-existence of a maximally entangled two-qubit state for a fixed spectrum in the non-pure rank-deficient cases. In Sec.\ IV we present the aforementioned results based on NE convertibility. We conclude with some discussion on the obtained results in Sec.\ V.

\section{Preliminaries}

In this article we only consider two-qubit quantum systems. Hence, their corresponding Hilbert space is $\mathbb{C}^2\otimes\mathbb{C}^2\simeq\mathbb{C}^4$. Thus, any element $|\psi\rangle$ of this space will always be given with respect to the canonical basis $\{|i\rangle\otimes|j\rangle:=|ij\rangle\}$, i.e.,
\begin{equation}
|\psi\rangle=\sum_{i,j=0}^1\psi_{ij}|ij\rangle,
\end{equation}
and analogously for any $4\times 4$ matrix $X$, i.e.\
\begin{equation}
X=\sum_{i,j,k,l=0}^1X_{ijkl}|ij\rangle\langle kl|.
\end{equation}
The set of all two-qubit states is characterized by the set of two-qubit density matrices, which is denoted by $\mathcal{D}$, i.e.,
\begin{equation}
\mathcal{D}=\{\rho\in\mathbb{C}^{4\times4}:\rho^\dag=\rho,\,\rho\geq0,\,\tr\rho=1\}.
\end{equation}
The range, kernel and support of a matrix will be denoted respectively by $R$, $\ker$ and $\supp$. Since all density matrices are Hermitian, for every $\rho\in\mathcal{D}$ it holds that $R(\rho)=\supp(\rho):=(\ker(\rho))^\perp$. Moreover, given any pure-state ensemble decomposition $\{p_i,|\psi_i\rangle\}_{i=1}^n$ of $\rho\in\mathcal{D}$, i.e.,
\begin{equation}
\rho=\sum_{i=1}^np_i|\psi_i\rangle\langle\psi_i|
\end{equation}
with $p_i>0$ $\forall i$ and $\sum_ip_i=1$, it is well-known that it must hold that $|\psi_i\rangle\in R(\rho)$ $\forall i$ (see e.g.\ \cite{PHorodecki97}). We will use $\mathcal{S}\subset\mathcal{D}$ to refer to the set of separable two-qubit density matrices, that is, density matrices $\rho\in\mathcal{D}$ additionally fulfilling that they can be written as
\begin{equation}
\rho=\sum_ip_i|\phi_i\rangle\langle\phi_i|\otimes|\chi_i\rangle\langle\chi_i|,
\end{equation}
for some choice of convex weights $\{p_i\}$ and unit-norm vectors $|\phi_i\rangle,|\chi_i\rangle\in\mathbb{C}^2$ $\forall i$. This set can be characterized by the partial transposition criterion \cite{ppt1,ppt2}. It states that $\rho\in\mathcal{S}$ if and only if (iff) $\rho^\Gamma\geq0$, where the superscript $\Gamma$ stands for the image under partial transposition, i.e., the linear map defined by the following action on the computational basis: $(|ij\rangle\langle kl|)^\Gamma=|kj\rangle\langle il|$. A state that is not separable is said to be entangled. A pure state is described by a rank-1 density matrix and it can be characterized by a (not necessarily normalized) choice of element in $\mathbb{C}^2\otimes\mathbb{C}^2$ that spans its support. In this case $\rho\propto|\psi\rangle\langle\psi|$ and notice that $\rho$ is separable iff $|\psi\rangle=|\phi\rangle\otimes|\chi\rangle$ for some choice of $|\phi\rangle,|\chi\rangle\in\mathbb{C}^2$. Thus, we label accordingly the (not necessarily normalized) elements of $\mathbb{C}^2\otimes\mathbb{C}^2$ as separable or entangled.   

State transformations in quantum theory are given by completely positive and trace-preserving (CPTP) maps. Any such map $\Lambda:\mathcal{D}\to\mathcal{D}$ admits a so-called Kraus representation,
\begin{equation}\label{Kraus}
\Lambda(\cdot)=\sum_{i=1}^kK_i\cdot K_i^\dag,
\end{equation}
where $\{K_i\}\subset\mathbb{C}^{4\times4}$ satisfying $\sum_iK^\dag_iK_i=\one$. For these and other facts about the theory of CPTP maps see e.g.\ \cite{watrous}. In the following we will consider particular classes of CPTP maps. The class of LOCC maps is notoriously involved to define and the reader is referred to \cite{locc}. Here we will only mention that local-unitary (LU) transformations, i.e., $\Lambda:\mathcal{D}\to\mathcal{D}$ such that
\begin{equation}
\Lambda(\cdot)=U_A\otimes U_B\cdot U_A^\dag\otimes U_B^\dag
\end{equation}
with $U_A,U_B\in\mathbb{C}^{2\times2}$ unitary matrices, are always LOCC. Furthermore, notice that these transformations are always invertible by another LU transformation, and, hence, by LOCC. Therefore, states related by local unitaries are always interconvertible by LOCC and, consequently, equivalent from the point of view of entanglement theory. Hence, whenever we speak about the set of entangled states and we make claims such as a given state being the unique entangled state having some property, it should be understood that we are speaking about the corresponding equivalence classes under this equivalence relation. It is known that the set of LOCC maps is a strict subset of the set of SEP maps, which, in turn, is a strict subset of the set of NE maps (see e.g.\ \cite{locc}). A CPTP map $\Lambda:\mathcal{D}\to\mathcal{D}$ is in the class of SEP maps if it admits a Kraus representation as in Eq.\ (\ref{Kraus}) such that $K_i=A_i\otimes B_i$ $\forall i$ for some matrices $\{A_i\},\{B_i\}\subset\mathbb{C}^{2\times2}$. A CPTP map $\Lambda:\mathcal{D}\to\mathcal{D}$ is NE if $\Lambda(\rho)\in\mathcal{S}$ $\forall\rho\in\mathcal{S}$. These two sets of maps are topologically closed and, therefore, if $\rho\in\mathcal{D}$ can be transformed to arbitrary precision to a state $\sigma\in\mathcal{D}$ by SEP (NE) maps, then there must exist a SEP (NE) map $\Lambda:\mathcal{D}\to\mathcal{D}$ such that $\Lambda(\rho)=\sigma$.

Given any possible spectrum for a density matrix in $\mathcal{D}$, we will denote by $\vec{\lambda}=(\lambda_1,\lambda_2,\lambda_3,\lambda_4)$ the corresponding 4-tuple of eigenvalues arranged in non-increasing order (which must then satisfy $\sum_i\lambda_i=1$ and $\lambda_i\geq0$ $\forall i$) and, as discussed in the introduction,
\begin{equation}
\s=\{\rho\in\mathcal{D}:\eig(\rho)=\{\lambda_1,\lambda_2,\lambda_3,\lambda_4\}\},
\end{equation}  
where $\eig(\cdot)$ stands for the spectrum of a matrix. The elements of the Bell basis of $\mathbb{C}^2\otimes\mathbb{C}^2$ are denoted by
\begin{align}\label{Bellvector}
|\Phi_1\rangle=\frac{1}{\sqrt{2}}(|00\rangle+|11\rangle)&,\quad |\Phi_2\rangle=\frac{1}{\sqrt{2}}(|00\rangle-|11\rangle),\nonumber\\
|\Phi_3\rangle=\frac{1}{\sqrt{2}}(|10\rangle+|01\rangle)&,\quad |\Phi_4\rangle=\frac{1}{\sqrt{2}}(|10\rangle-|01\rangle),
\end{align}
and the corresponding density matrices by~$\Phi_i=|\Phi_i\rangle\langle\Phi_i|$ ($i\in\{1,2,3,4\}$). The MEMSs of \cite{mems1} that maximize several entanglement measures in $\s$ are then given by
\begin{equation}\label{mems}
\rho_{\vec{\lambda}}=\lambda_1\Phi_1+\lambda_2|01\rangle\langle01|+\lambda_3\Phi_2+\lambda_4|10\rangle\langle10|:=\sum_{i=1}^4\lambda_i\xi_i,
\end{equation}
where we use the shorthand notation
\begin{equation}
\xi_1=\Phi_1,\quad\xi_2=|01\rangle\langle01|,\quad\xi_3=\Phi_2,\quad\xi_4=|10\rangle\langle10|.
\end{equation}
As explained in the introduction, if there exists a choice of $\vec{\lambda}$ and $\sigma\in\s$ such that $\rho_{\vec{\lambda}}$ cannot be transformed by SEP or NE into $\sigma$, then we can conclude that there is no maximally entangled state in $\s$. We will often consider states in $\mathcal{D}$ that are diagonal in the Bell basis, i.e., 
\begin{equation}\label{Belldiagonal}
\sigma=\sum_{i=1}^4p_i\Phi_i,
\end{equation}
where $\sum_ip_i=1$ and $p_i\geq0$ $\forall i$. These states have $\{p_i\}$ as their spectrum and they are known to be entangled iff $\max_ip_i>1/2$ \cite{belldiagonalstates}. In fact, it is known \cite{isotropic} that $\tr(\rho\Phi_i)\leq1/2$ must hold for any $i\in\{1,2,3,4\}$ if $\rho\in\mathcal{S}$. In particular, the main result of \cite{nomems} is that there exists no NE map transforming $\rho_{(\lambda,1-\lambda,0,0)}$ into $\lambda\Phi_1+(1-\lambda)\Phi_3$ when $\lambda\in(2/3,1)$, implying that there is no maximally entangled state in $S(\lambda,1-\lambda,0,0)$ under the above hypothesis on the parameter $\lambda$. 

It should be noted that it follows from the results of \cite{mems1} that all elements of $\s$ are separable iff 
\begin{equation}\label{absolutesep}
\lambda_1 - \lambda_3 - 2 \sqrt{\lambda_2 \lambda_4} \leq 0.
\end{equation}
Thus, the question of the existence of a maximally entangled state in $\s$ only makes sense under the assumption that the above inequality is violated. This is always the case when the rank equals 2. The same happens when the rank equals 3 except in the case $\lambda_1=\lambda_3$ (i.e., $\vec{\lambda}=(1/3,1/3,1/3,0)$).

\section{Impossibility results via SEP transformations}

To study SEP transformations among rank-deficient states in $\mathcal{D}$, we will use certain properties of strict subspaces of $\mathbb{C}^2\otimes\mathbb{C}^2$ that will correspond to the supports of the involved density matrices (which, as mentioned before, is the same as their ranges). In particular, notice first that if $V$ is a 2-dimensional subspace of $\mathbb{C}^{2}\otimes\mathbb{C}^2$ then either all states in $V$ are separable or there is only one separable state or there are only 2 separable states \cite{vidalsanpera}. If $V=\textrm{span}\{|\chi_1\rangle,\ldots,|\chi_m\rangle\}$ and $A,B\in\mathbb{C}^{2\times2}$, we denote by $(A\otimes B)V:=\textrm{span}\{A\otimes B|\chi_1\rangle,\ldots,A\otimes B|\chi_m\rangle\}$.

\begin{lemma}\label{subspace1}
If $V$ is a 2-dimensional subspace of $\mathbb{C}^{2}\otimes\mathbb{C}^2$ and $A,B\in\mathbb{C}^{2\times2}$ are invertible, then $(A\otimes B)V$ is a 2-dimensional subspace that has the property of having all states separable or just one or just two iff $V$ has the same property.
\end{lemma}
\begin{proof}
This is obvious from the fact that $A$ and $B$ are invertible (and, hence, so is $A\otimes B$ and $(A\otimes B)V$ must be 2-dimensional as well) and that in this case $|\psi\rangle\in\mathbb{C}^{2}\otimes\mathbb{C}^2$ is separable iff so is $A\otimes B|\psi\rangle$.
\end{proof}

Three-dimensional subspaces of $\mathbb{C}^{2}\otimes\mathbb{C}^2$ can be catalogued by whether their (one-dimensional) orthogonal complement is spanned by a separable or an entangled state. We have a similar property as in the previous case.
\begin{lemma}\label{subspace2}
If $V$ is a 3-dimensional subspace of $\mathbb{C}^{2}\otimes\mathbb{C}^2$ and $A,B\in\mathbb{C}^{2\times2}$ are invertible, then $(A\otimes B)V$ is a 3-dimensional subspace that has the property of having a separable or an entangled orthogonal complement iff $V$ has the same property.
\end{lemma}
\begin{proof}
The proof follows the same reasoning as the previous lemma by noticing that if $|\psi\rangle\in V^\perp$, then $A^{-1}\otimes B^{-1}|\psi\rangle\in ((A\otimes B)V)^\perp$.
\end{proof}

We can now state and prove the main results in this section. Given the different structures described above, we consider the cases of rank equal to two and rank equal to three separately. However, the main idea behind the proofs is the same.

\begin{theorem}\label{th1SEP}
There is no maximally entangled state in $S(\lambda,1-\lambda,0,0)$ for all $\lambda\in[1/2,1)$.
\end{theorem}
\begin{proof}
Let $\rho_\lambda=\rho_{(\lambda,1-\lambda,0,0)}$ (cf.\ Eq.\ (\ref{mems})). As already explained, it is enough to prove that for every given $\lambda\in[1/2,1)$ there exists no SEP map that transforms $\rho_\lambda$ into a particular choice of state in $S(\lambda,1-\lambda,0,0)$. We establish this choice in the following. For a convenient choice of $\epsilon\in\mathbb{R}$ to be specified later, we denote by $|\Phi_1(\epsilon)\rangle$ the properly normalized state proportional to $|\Phi_1\rangle+\epsilon|10\rangle$. Additionally, we use $\Phi_1(\epsilon)=|\Phi_1(\epsilon)\rangle\langle\Phi_1(\epsilon)|$ and define
\begin{equation}
\rho_\lambda(\epsilon)=\lambda\Phi_1(\epsilon)+(1-\lambda)|01\rangle\langle01|.
\end{equation}
Notice that $\rho_\lambda(\epsilon)$ is isospectral to $\rho_\lambda=\rho_\lambda(0)$ for every $\epsilon$ and that for every choice of $\lambda\in[1/2,1)$ there exists a choice of $\epsilon>0$ sufficiently small such that $\rho_\lambda(\epsilon)$ is entangled as well. This is due to the well-known fact that the set of separable states is topologically closed (see e.g.\ \cite{PHorodecki97}). Using such a choice of $\epsilon$ for any given $\lambda$ we prove the theorem by showing that $\rho_\lambda$ cannot be transformed to $\rho_\lambda(\epsilon)$ by SEP. Assume then for a contradiction that there exists a SEP map $\Lambda:\mathcal{D}\to\mathcal{D}$ such that 
\begin{equation}\label{th1SEP:assumption}
\Lambda(\rho_\lambda)=\rho_\lambda(\epsilon)
\end{equation}
and let $V_\epsilon=\textrm{span}\{|\Phi_1(\epsilon)\rangle,|01\rangle\}=R(\rho_\lambda(\epsilon))$, which has the property of containing exactly two separable states for all $\epsilon>0$ (see e.g.\ \cite{slocc}). On the other hand, $V_0=\textrm{span}\{|\Phi_1\rangle,|01\rangle\}=R(\rho_\lambda)$ contains exactly one separable state \cite{slocc}. By linearity, Eq.\ (\ref{th1SEP:assumption}) enforces that
\begin{equation}
\lambda\Lambda(\Phi_1)+(1-\lambda)\Lambda(|01\rangle\langle01|)=\lambda\Phi_1(\epsilon)+(1-\lambda)|01\rangle\langle01|.
\end{equation}
Since  $\Lambda(\Phi_1)$ and $\Lambda(|01\rangle\langle01|)$ are positive semidefinite, this implies that  $V_\epsilon^\perp\subseteq \textrm{ker}(\Lambda(\Phi_1))$ and $V_\epsilon^\perp\subseteq\textrm{ker}(\Lambda(|01\rangle\langle01|))$. Consequently, we have that $R(\Lambda(\Phi_1))\subseteq V_\epsilon$ and $R(\Lambda(|01\rangle\langle01|))\subseteq V_\epsilon$. If there exists a SEP map $\Lambda(\cdot)=\sum_iA_i\otimes B_i\cdot A_i^\dagger\otimes B_i^\dagger$ fulfilling the last condition, then
\begin{equation}
\left\{|\psi_i\rangle:=\frac{A_i\otimes B_i|\Phi_1\rangle}{||A_i\otimes B_i|\Phi_1\rangle||}\right\},\quad \left\{|\phi_i\rangle:=\frac{A_i\otimes B_i|01\rangle}{||A_i\otimes B_i|01\rangle||}\right\}
\end{equation}
give pure states respectively in ensemble decompositions of $\Lambda(\Phi_1)$ and $\Lambda(|01\rangle\langle01|)$, and, therefore, they all must belong to their respective ranges and, hence, to $V_\epsilon$ (whenever the above norms are non-zero). Moreover, given that $\rho_\lambda(\epsilon)$ is entangled, so must be $\Lambda(\Phi_1)$ and there must exist at least one value of $i$ (take without loss of generality $i=1$) such that $A_1$ and $B_1$ are invertible and $|\psi_1\rangle\neq0$ is an entangled state. The above invertibility property guarantees that $|\phi_1\rangle\neq0$, which must be separable. Thus, $|\psi_1\rangle$ and $|\phi_1\rangle$ are two different (and linearly independent) vectors in $V_\epsilon$ and, consequently, $V_\epsilon=\textrm{span}\{|\psi_1\rangle, |\phi_1\rangle\}$. However, this means that $V_\epsilon=(A_1\otimes B_1)V_0$ and, as per Lemma \ref{subspace1}, we have reached a contradiction.
\end{proof} 

\begin{theorem}\label{th2SEP}
There is no maximally entangled state in $S(\lambda_1,\lambda_2,\lambda_3,0)$ for any 4-tuple of ordered eigenvalues $\vec\lambda=(\lambda_1,\lambda_2,\lambda_3,0)$ with $\lambda_1\neq\lambda_3>0$ (cf.\ Eq.\ (\ref{absolutesep})).
\end{theorem}
\begin{proof}
The proof is very similar to the case of rank equal to 2 and we use the same notation as therein, where now 
\begin{equation}
\rho_{\vec\lambda}(\epsilon)=\lambda_1\Phi_1(\epsilon)+\lambda_2|01\rangle\langle01|+\lambda_3\Phi_2,
\end{equation}
which still has the property of being isospectral to $\rho_{\vec\lambda}=\rho_{\vec\lambda}(0)$. For the same reasons as before, for every choice of $\vec\lambda$ such that $\rho_{\vec\lambda}$ is entangled (i.e., $\lambda_1\neq\lambda_3$), there exists a choice of $\epsilon>0$ sufficiently small such that $\rho_{\vec\lambda}(\epsilon)$ is entangled as well. For such a choice of $\epsilon$ for any given $\vec\lambda$, assume then for a contradiction that there exists a SEP map $\Lambda:\mathcal{D}\to\mathcal{D}$ such that 
\begin{equation}
\Lambda(\rho_{\vec\lambda})=\rho_{\vec\lambda}(\epsilon)
\end{equation}
and let $V_\epsilon=\textrm{span}\{|\Phi_1(\epsilon)\rangle,|01\rangle,|\Phi_2\rangle\}=R(\rho_{\vec\lambda}(\epsilon))$, which has the property of having an entangled orthogonal complement for all $\epsilon>0$ (spanned by $\epsilon|00\rangle+\epsilon|11\rangle-\sqrt{2}|10\rangle$)\footnote{To see this, notice that $$\epsilon|00\rangle+\epsilon|11\rangle-\sqrt{2}|10\rangle=(A_\epsilon\otimes\one)(|00\rangle+|11\rangle)$$ with $$A_\epsilon=\left(
                    \begin{array}{cc}
                      \epsilon & 0 \\
                      -\sqrt{2} & \epsilon \\
                    \end{array}
                  \right).$$ Hence, $A_\epsilon$ is invertible $\forall\epsilon>0$.}. On the other hand, the orthogonal complement of $V_0=\textrm{span}\{|\Phi_1\rangle,|01\rangle,|\Phi_2\rangle\}=R(\rho_{\vec\lambda})$ is spanned by the separable state $|10\rangle$. By linearity of $\Lambda$, like before, it must hold that $R(\Lambda(\Phi_1)),R(\Lambda(|01\rangle\langle01|)),R(\Lambda(\Phi_2))\subseteq V_\epsilon$ and if there exists a SEP map $\Lambda(\cdot)=\sum_iA_i\otimes B_i\cdot A_i^\dagger\otimes B_i^\dagger$ fulfilling the above condition, then
\begin{align}
&\left\{|\psi_i\rangle:=\frac{A_i\otimes B_i|\Phi_1\rangle}{||A_i\otimes B_i|\Phi_1\rangle||}\right\},\nonumber\\
&\left\{|\phi_i\rangle:=\frac{A_i\otimes B_i|01\rangle}{||A_i\otimes B_i|01\rangle||}\right\},\nonumber\\
&\left\{|\chi_i\rangle:=\frac{A_i\otimes B_i|\Phi_2\rangle}{||A_i\otimes B_i|\Phi_2\rangle||}\right\}
\end{align}
must all belong to $V_\epsilon$ (whenever the above norms are non-zero). Furthermore, given that $\rho_\lambda(\epsilon)$ is entangled, there must exist at least one value of $i$ (take without loss of generality $i=1$) such that $A_1$ and $B_1$ are invertible and, hence, $|\psi_1\rangle$, $|\phi_1\rangle$ and $|\chi_1\rangle$ are nonzero and moreover linearly independent (since so are $\{|\Phi_1\rangle,|01\rangle,|\Phi_2\rangle\}$). This entails that $V_\epsilon=\textrm{span}\{|\psi_1\rangle, |\phi_1\rangle,|\chi_1\rangle\}$. However, this means that $V_\epsilon=(A_1\otimes B_1)
V_0$ and, as per Lemma \ref{subspace2}, we have reached a contradiction.
\end{proof}

\section{Impossibility results via NE transformations}

The techniques of the previous section rely heavily on the fact that both the input and target density matrices of the SEP protocol have non-trivial kernels and, thus, it is not clear how to use them in the case of full-rank density matrices. Moreover, as already discussed in the introduction, even if we have already proven that a maximally entangled state for a fixed spectrum never exists in the non-pure rank-deficient case in the standard resource theory of entanglement that uses LOCC convertibility, it is still interesting to explore whether this restriction can be lifted by considering a resource theory with a more permissive class of transformations. This leads us to consider NE maps, which is the largest class of transformations that give rise to a well-defined resource theory of entanglement for single-copy manipulation \cite{resource}. Reference \cite{nomems} established the non-existence of a maximally entangled state in $S(\lambda,1-\lambda,0,0)$ for $\lambda\in(2/3,1)$ by proving that there exists no NE map in this case transforming $\rho_{\vec{\lambda}}$ into an isospectral rank-2 Bell-diagonal state (cf.\ Eq.\ (\ref{Belldiagonal})). Here, we extend this approach to arbitrary rank and consider again the possibility of transforming $\rho_{\vec{\lambda}}$ to an isospectral Bell-diagonal state, which from now on we are going to denote by  
\begin{eqnarray}\label{eq:sigma}
\sigma_{\vec{\lambda}}=\sum_j \lambda_j \Phi_j.
\end{eqnarray}
Thus, we aim at finding obstructions to the existence of a NE map $\Lambda:\mathcal{D}\to\mathcal{D}$ such that $\Lambda(\rho_{\vec{\lambda}})=\sigma_{\vec{\lambda}}$. As already commented in the introduction, we are going to do so by reducing this question to the feasibility of a linear program (LP). 

Before addressing this, some comments are in order. First, as discussed after Eq.\ (\ref{Belldiagonal}), $\sigma_{\vec{\lambda}}$ is entangled iff $\lambda_1>1/2$. Notice that this only defines a strict subset of the region corresponding to inverting the inequality in Eq.\ (\ref{absolutesep}), where the question of the existence of a maximally entangled state in $\s$ is well posed. Therefore, since a transformation by LOCC from an entangled state to a separable state is always possible, our approach carries the inherent limitation that it will never work to discard the existence of a maximally entangled state in $\s$ whenever $\lambda_1\leq1/2$ (see actually Theorem \ref{thsufNE} below). Second, and in relation to this, it thus comes as a natural option to consider target states that are not Bell-diagonal. In fact, the LP approach that we are going to present can be immediately adapted to consider arbitrary isospectral target states. However, we have numerically investigated this option with random target states and we have found that this does not lead to any substantial improvement. Last, even if we restrict ourselves to Bell-diagonal states as targets, there is still some freedom in how to assign the eigenvalues to the Bell eigenstates (recall that the eigenvalues are ordered). Nevertheless, once we explain the LP approach, it will become apparent that this has no effect in its feasibility. Consequently, without loss of generality in what comes to the performance of our method, we can stick to the assignment made in Eq.\ (\ref{eq:sigma}).

\subsection{Eigenvalue distributions for which a NE transformation to Bell-diagonal states is possible}

As we have just discussed, if we fix the target states to be Bell-diagonal states, we will never be able to prove the nonexistence of a maximally entangled state in $\s$ for spectra satisfying $\lambda_1\leq1/2$. It turns out that the situation is more complex once we restrict ourselves to NE transformations. As we prove in the following, the MEMS state $\rho_{\vec{\lambda}}$ happens to be convertible by NE transformations to isospectral Bell-diagonal states for certain non-trivial eigenvalue distributions where $\sigma_{\vec{\lambda}}$ is actually entangled.

\begin{theorem}\label{thsufNE}
If it holds that 
\begin{equation}\label{thsufNE:assumption}
1-\lambda_1-2\lambda_2\geq0,
\end{equation} 
then there exists a NE map $\Lambda:\mathcal{D}\to\mathcal{D}$ such that $\Lambda(\rho_{\vec{\lambda}})=\sigma_{\vec{\lambda}}$.
\end{theorem}
\begin{proof}
We construct explicitly the NE map $\Lambda$ that implements the transformation under the stated hypothesis on the spectrum. Namely, since the case $\lambda_1=1$ is trivial, let us assume that $\lambda_1\neq1$ and
\begin{equation}
\Lambda(X)=\tr(\Phi_1X)\Phi_1+\frac{\tr[(\one-\Phi_1)X]}{1-\lambda_1}(\lambda_2\Phi_2+\lambda_3\Phi_3+\lambda_4\Phi_4).
\end{equation}
The map $\Lambda$ is manifestly CPTP because it has the form of a measure-and-prepare quantum channel (see e.g.\ \cite{watrous}) and it obviously has the desired property that $\Lambda(\rho_{\vec{\lambda}})=\sigma_{\vec{\lambda}}$. Hence, it only remains to prove that $\Lambda$ is NE. Notice first that the output of the map is a convex combination of two states and that the latter is separable. This is because the condition given by Eq.\ (\ref{thsufNE:assumption}) is equivalent to
\begin{equation}
\frac{\lambda_2}{1-\lambda_1}\leq\frac{1}{2}.
\end{equation}
Thus, 
\begin{equation}
\frac{\lambda_2\Phi_2+\lambda_3\Phi_3+\lambda_4\Phi_4}{1-\lambda_1}
\end{equation}
is a Bell-diagonal state with maximal eigenvalue not larger than $1/2$, and, hence, separable. As a consequence, for any $\rho\in\mathcal{D}$, $\Lambda(\rho)$, which is also a Bell-diagonal state, can only be entangled if $\tr(\Phi_1\rho)>1/2$. But, for every $\rho\in\mathcal{S}$ it must hold that $\tr(\Phi_1\rho)\leq1/2$. Therefore, it must hold that $\Lambda(\rho)\in\mathcal{S}$ for every $\rho\in\mathcal{S}$, as we wanted to prove.
\end{proof}

Notice that this does not of course imply that a maximally entangled state exists in $\s$ when the condition of Eq.\ (\ref{thsufNE:assumption}) holds. It could still be that it is impossible to transform $\rho_{\vec{\lambda}}$ to other isospectral states in this region. Notice as well that Eq.~\eqref{thsufNE:assumption} is only meaningful when the rank is 3 or 4, because otherwise it reduces to $\lambda_1-1\geq 0$, which can only occur in the trivial case of rank equal to 1. Moreover, when the rank is 3 it can only hold when $\lambda_2=\lambda_3$, and, in particular, it is worth pointing out that in this case the transformation is possible by NE but not by SEP. This is because when $\lambda_3>0$ $(\supp(\rho_{(\lambda_1,\lambda_2,\lambda_3,0)}))^\perp=\ker(\rho_{(\lambda_1,\lambda_2,\lambda_3,0)})$ is spanned by a separable state, while $\ker(\sigma_{(\lambda_1,\lambda_2,\lambda_3,0)})$ is spanned by an entangled state. Thus, we can use the exact same arguments as in the proof of Theorem \ref{th2SEP} to see that a rank-3 MEMS given by $\rho_{\vec{\lambda}}$ can never be transformed by SEP into a (isospectral or not) rank-3 Bell-diagonal state.

\subsection{Eigenvalue distributions for which a NE transformation to Bell-diagonal states is impossible}

Along the lines of~\cite{nomems}, by assuming the existence of a NE map $\Lambda(\rho_{\vec{\lambda}})=\sigma_{\vec{\lambda}}$ for fixed $\vec{\lambda}$ one can determine a set of fundamental constrains that must be satisfied; if this is not the case, $\vec{\lambda}$ can be ruled out as a valid spectral point where the map exists. As already mentioned, this constraints can be cast as a feasibility question for a LP. In general, the latter has the following form
\begin{eqnarray}\label{eq:lp}
\max\limits_{x} \hspace{2pt} & c^Tx &\nonumber\\
\text{such that} \hspace{2pt} Ax &\leq & b,\nonumber\\
x &\geq 0&,
\end{eqnarray}
where $A\in \mathbb{R}^{n\times m}$, $x,c\in \mathbb{R}^m$, and $b\in \mathbb{R}^n$ are given. The set $\{x\in\mathbb{R}^m:Ax\leq b,x\geq0\}$ is called the feasible region of the LP and, by construction, it is always a convex polytope. When the feasible set is empty, the LP is said to be infeasible. There are two reasons why a given LP might not admit a solution. One obvious possibility is because the LP is infeasible. The other is because the feasible region is unbounded (which will never happen in our case because in our problem it will also hold that $x\leq(1,1,\ldots,1)$). LPs constitute a basic field of study in optimization theory. There are many algorithms, such as interior point methods, that allow to solve them efficiently and they are included in standard libraries for most programming languages. Furthermore, duality theory provides means to certify that a numerically found solution is indeed correct or that the LP is infeasible. For more details regarding LPs the reader is referred to \cite{boyd} and references therein.

In order to present this LP approach, we begin by assuming that there exists a NE map $\Lambda$ such that $\Lambda(\rho_{\vec{\lambda}})=\sigma_{\vec{\lambda}}$. Recalling Eq.\ (\ref{mems}), by linearity this implies that
\begin{equation}\label{eq:ne1}
\sum_j\lambda_j\Lambda(\xi_j)=\sum_j\lambda_j\Phi_j.
\end{equation}
At this point, we define
\begin{eqnarray}
F_{ij}=\tr(\Phi_i\Lambda(\xi_j))
\end{eqnarray}
and from Eq.~\eqref{eq:ne1}, it follows that
\begin{equation}\label{1}
\sum_j\lambda_jF_{ij}=\lambda_i
\end{equation}
holds for all $i\in\{1,2,3,4\}$. Since $\Lambda$ is trace-preserving, it must also hold that
\begin{equation}\label{2}
\sum_{i}F_{ij}=1\quad\forall j.
\end{equation}
Given that $\Lambda$ is a positive map, we also have the obvious inequality constraints
\begin{equation}\label{3}
F_{ij}\geq0\quad\forall i,j.
\end{equation}
Additionally, because $(\xi_1+\xi_3)/2$ is separable we must have as well that
\begin{equation}\label{4}
F_{i1}+F_{i3}\leq1\quad\forall i,
\end{equation}
and because $\xi_2$ and $\xi_4$ are separable that
\begin{equation}\label{5}
F_{i2}\leq\frac{1}{2},\quad F_{i4}\leq\frac{1}{2}\quad\forall i.
\end{equation}
Notice that for $i=1$ and $i=3$
\begin{equation}
\frac{1}{2}\xi_i+\frac{1}{4}(\xi_2+\xi_4)
\end{equation}
is a separable Bell-diagonal state. Therefore, we also obtain the constraints
\begin{equation}\label{6}  
F_{i1}+\frac{1}{2}(F_{i2}+F_{i4})\leq1\quad\forall i
\end{equation}
and
\begin{equation}\label{7}
F_{i3}+\frac{1}{2}(F_{i2}+F_{i4})\leq1\quad\forall i.
\end{equation}

Actually, the conditions given by Eqs.\ (\ref{5}), (\ref{6}) and (\ref{7}) can be generalized by considering for $a\in[0,1]$ the family of states given by
\begin{eqnarray}\label{eq:tau_states}
\tau_a^\pm &=& 2a(1-a)\Phi_1+a^2|01\rangle\langle01|+(1-a)^2|10\rangle\langle10|
\nonumber\\
&\pm & a(1-a)(|01\rangle\langle10|+|10\rangle\langle01|)\nonumber\\
\tilde{\tau}_a^\pm &=& 2a(1-a)\Phi_2+a^2|01\rangle\langle01|+(1-a)^2|10\rangle\langle10|\nonumber\\
&\pm & a(1-a)(|01\rangle\langle10|+|10\rangle\langle01|).
\end{eqnarray}
It can be easily checked that $(\tau_a^\pm)^\Gamma\geq0$ and $(\tilde{\tau}_a^\pm)^\Gamma\geq0$ hold $\forall a\in[0,1]$. Consequently, these states are always separable and it must therefore hold for arbitrary values of $a\in[0,1]$ and for all $i$ that
\begin{eqnarray}\label{8}
2a(1-a)F_{i1}+a^2 F_{i2} +(1-a)^2 F_{i4} &\leq & \frac{1}{2},\nonumber\\
2a(1-a)F_{i3}+a^2 F_{i2} +(1-a)^2 F_{i4}&\leq & \frac{1}{2}.
\end{eqnarray}
Notice that the aforementioned inequalities correspond to the cases $a\in\{0,1/2,1\}$.

In summary, for a given spectrum distribution $\vec{\lambda}$, if there exists a NE map $\Lambda$ such that $\Lambda(\rho_{\vec{\lambda}})=\sigma_{\vec{\lambda}}$, then it must be possible to assign values to the 16 variables $\{F_{ij}\}$ such that the conditions given by Eqs.\ \eqref{1}, \eqref{2}, \eqref{3}, \eqref{4}, and \eqref{8} are all fulfilled at the same time. By considering a discretization of the interval $[0,1]$ in which the parameter $a$ in Eq.\ (\ref{8}) takes values, we obtain a finite list of linear constraints. Thus, if we arrange the unknowns $\{F_{ij}\}$ into a vector $x\in\mathbb{R}^{16}$, for which Eq.~\eqref{3} becomes $x\geq0$, the problem can be expressed in the form $Ax\leq b$ for some suitably chosen matrix $A$ and vector $b$ that only depend on $\vec{\lambda}$. Hence, our problem boils down to the feasibility of a given LP. If for a given spectrum $\vec{\lambda}$ this LP is infeasible, then we can conclude that there exists no NE map implementing the transformation $\Lambda(\rho_{\vec{\lambda}})=\sigma_{\vec{\lambda}}$, and, as a consequence, that the set $\s$ does not admit a maximally entangled state.

%%%%%%%%%%%%%%%%%%%%%%%%%%%%%%%%%%%%%%%%
\begin{figure}[!t]
\centering
\includegraphics[scale=0.63]{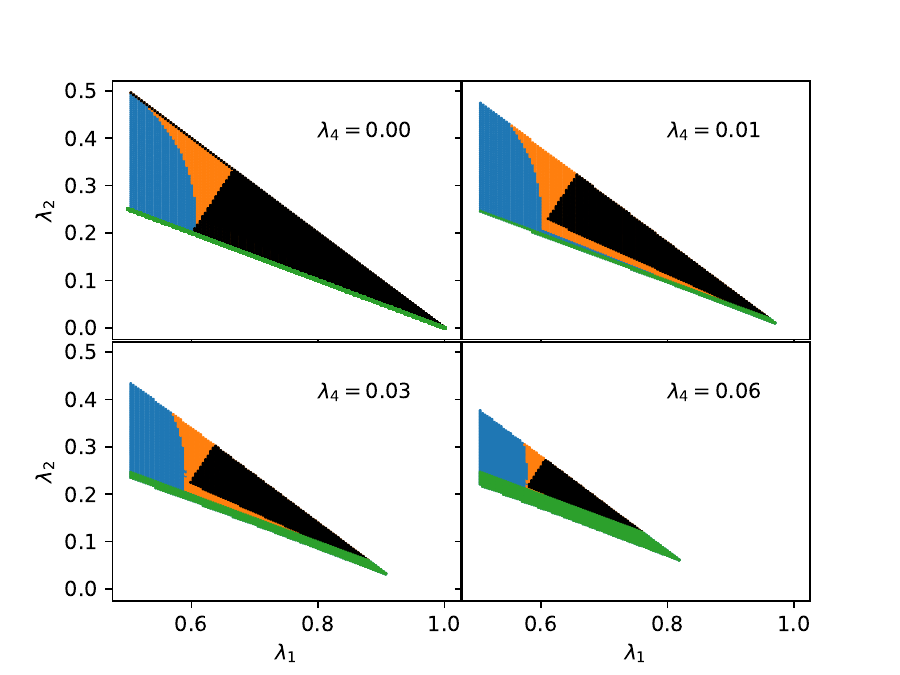}   
    \caption{Feasibility of the LP for values of $\vec{\lambda}$ constrained to $\lambda_1>1/2$ for the choices $\lambda_4=0,0.01,0.03,0.06$ (notice that $\lambda_1>\lambda_3+2\sqrt{\lambda_2\lambda_4}$ is always satisfied in these cases). Note that $\lambda_4=0$ corresponds to the rank-3 case. Feasible regions of the LP are in blue and green and infeasible regions are in orange and black. The green region corresponds to the points where we analytically know that the LP is feasible due to Theorem~\ref{thsufNE}, while the black region corresponds to the points of infeasibility due to Theorems~\ref{th:thne1}, \ref{th:thne2} and \ref{th:thne4}.}
    \label{fig:fig1}
\end{figure}
%%%%%%%%%%%%%%%%%%%%%%%%%%%%%%%%%%%%%%%

The question of the feasibility of this LP can be easily solved with the aid of a computer. Figure \ref{fig:fig1} depicts the results obtained for four different choices of $\lambda_4$, which were obtained in a few seconds by an implementation in a standard desktop computer. The green region corresponds to the points where we not only know that the LP is feasible but that the transformation is indeed possible by NE as per Theorem \ref{thsufNE}. Other feasible regions of the LP have been colored in blue, whereas infeasible regions have been colored in orange. Thus, for all eigenvalue distributions corresponding to this latter case a maximally entangled state for the given spectrum does not exist. As can be seen from Fig. \ref{fig:fig1}, our approach allows to prove this property for a general set of possible spectra within the constraints given by $\lambda_1>1/2$ and that of Theorem~\ref{thsufNE} even in the case of full rank and under the most permissive class of transformations. The black region within the orange region, corresponds to the cases where we can rule out the feasibility of the LP without the need of numerical means, which is the content of the remaining of this section. While this region is strictly contained in the orange region and we have therefore not been able to prove analytically the infeasibility of the LP in the exact same regions that numerics dictate, the following results allow us to establish the nonexistence of a maximally entangled state for a fixed spectrum under NE transformations for a large class of eigenvalue distributions and not just a specific choice, which is as far as numerical implementations can go.  

Before concluding this section with these last results, we are now in the position to explain why the feasibility of the LP does not depend on how we choose to match the eigenvalues with the Bell eigenstates in Eq.\ (\ref{eq:sigma}). Let $\pi$ be an arbitrary permutation in $\{1,2,3,4\}$ and suppose that we choose instead as a target state
\begin{equation}   
\sigma^\pi_{\vec{\lambda}}=\sum_j \lambda_{\pi(j)} \Phi_j.
\end{equation}
Under the assumption that there is a NE map $\Lambda$ such that $\Lambda(\rho_{\vec{\lambda}})=\sigma^\pi_{\vec{\lambda}}$, all the constraints from before stay the same except for that of Eq.\ (\ref{1}), which now reads
\begin{equation}
\sum_j\lambda_jF_{ij}=\lambda_{\pi(i)}.
\end{equation} 
Hence, the assignment $\{F_{ij}\}$ is feasible for the LP corresponding to $\Lambda(\rho_{\vec{\lambda}})=\sigma_{\vec{\lambda}}$ iff the assignment $\{F_{\pi(i)j}\}$ is feasible for the LP corresponding to $\Lambda(\rho_{\vec{\lambda}})=\sigma^\pi_{\vec{\lambda}}$. Therefore, one LP is feasible iff the other is.

\subsubsection{Rank-2 case}

\begin{theorem}\label{th:thne1}
Let $\rho_{\vec{\lambda}}$ and $\sigma_{\vec{\lambda}}$ be given respectively by Eqs.~\eqref{mems} and~\eqref{eq:sigma} with $\vec{\lambda}=(\lambda,1-\lambda,0,0)$. Then, there exists no NE map $\Lambda:\mathcal{D}\to\mathcal{D}$ such that $\Lambda(\rho_{\vec{\lambda}})=\sigma_{\vec{\lambda}}$ for $\lambda\in(1/2,1)$.
\end{theorem}   
\begin{proof}
Suppose such a map exists. Then, arguing as in \cite{nomems} under the assumption that $\lambda<1$ we obtain that it must hold that
\begin{equation}\label{condsprl}
F_{11}=\frac{3\lambda-1}{2\lambda},\quad F_{21}=\frac{1-\lambda}{2\lambda},\quad
F_{12}=F_{22}=\frac{1}{2}.
\end{equation}
Then, the condition of Eq.\ (\ref{8}) boils down to
\begin{equation}
2a(1-a)\frac{3\lambda-1}{2\lambda}+\frac{a^2}{2}\leq\frac{1}{2}.
\end{equation}
Thus, choosing $a=1/(2\lambda)\in(1/2,1)$, it follows that it must hold that
\begin{equation}
\frac{12\lambda^2-9\lambda+2}{8\lambda^3}\leq\frac{1}{2}
\end{equation}
which is a contradiction when $\lambda>1/2$. To see this, notice that the last inequality is equivalent to
\begin{equation}
4\lambda^3-12\lambda^2+9\lambda-2\geq0
\end{equation}
and that $4\lambda^3-12\lambda^2+9\lambda-2=4(\lambda-1/2)^2(\lambda-2)$.
\end{proof}

\subsubsection{Rank-3 case}

\begin{theorem}\label{th:thne2}
Let $\rho_{\vec{\lambda}}$ and $\sigma_{\vec{\lambda}}$ be given respectively by Eqs.\ \eqref{mems} and~\eqref{eq:sigma} with $\lambda_3>\lambda_4=0$. Then, there exists no NE map $\Lambda:\mathcal{D}\to\mathcal{D}$ such that $\Lambda(\rho_{\vec{\lambda}})=\sigma_{\vec{\lambda}}$ whenever $2\lambda_1-\lambda_2>1$ and $\lambda_2>\lambda_3$.
\end{theorem}
\begin{proof}
To prove the claim we assume that $2\lambda_1-\lambda_2>1$ and we will show that if there exists a NE map $\Lambda$ such that $\Lambda(\rho_{\vec{\lambda}})=\sigma_{\vec{\lambda}}$, then it must necessarily hold that $\lambda_2=\lambda_3$. Notice first that the premise that $\lambda_4=0$ implies that $F_{4j}=0$ $\forall j$. Equation (\ref{1}) when $i=1$ reads
\begin{equation}\label{principal1}
\lambda_1=\lambda_1F_{11}+\lambda_2F_{12}+(1-\lambda_1-\lambda_2)F_{13},
\end{equation}
which can be rewritten as
\begin{eqnarray}
\lambda_1&=&(1-\lambda_1-\lambda_2)(F_{11}+F_{13})\nonumber\\
&+&(2\lambda_1+\lambda_2-1)(F_{11}+F_{12}/2)\nonumber\\
&-&(2\lambda_1-\lambda_2-1)F_{12}/2.
\end{eqnarray}
Since $1-\lambda_1-\lambda_2\geq0$ and $2\lambda_1+\lambda_2-1=\lambda_1-\lambda_3\geq0$, Eqs. (\ref{4}) and (\ref{6}) together with the above one entail that
\begin{equation}
\lambda_1\leq\lambda_1-(2\lambda_1-\lambda_2-1)F_{12}/2,
\end{equation}
and, thus, our assumption that $2\lambda_1-\lambda_2>1$ enforces that $F_{12}=0$. This implies, by Eqs.\ (\ref{2}) and (\ref{5}) and the fact that $F_{42}=0$, that $F_{22}=F_{32}=1/2$. In addition to this, Eq.\ (\ref{principal1}) boils down to
\begin{equation}
\lambda_1=\lambda_1F_{11}+(1-\lambda_1-\lambda_2)F_{13},
\end{equation}
 and because of Eq.\ (\ref{4}) we must then conclude that $F_{11}=1$ and $F_{13}=0$. This leads then to $F_{21}=F_{31}=0$ as per Eq.\ (\ref{2}). With all these constraints Eq.\ (\ref{1}) for $i=2$ now reads
 \begin{equation}
 F_{23}=\frac{\lambda_2}{2(1-\lambda_1-\lambda_2)}=\frac{\lambda_2}{2\lambda_3}.
 \end{equation}
However, the second inequality in Eq.\ (\ref{8}) imposes then for $a<1$ that 
\begin{equation}
\frac{\lambda_2}{2\lambda_3}=F_{23}\leq\frac{1+a}{4a}\to_{a\to1}\frac{1}{2}.
\end{equation}
Thus, if $\lambda_2>\lambda_3$, we immediately obtain a contradiction in the above condition by taking $a$ sufficiently close to 1. Hence, the last inequality can only hold if $\lambda_2=\lambda_3$, as we wanted to show.
\end{proof}

\subsubsection{Rank-4 case}

\begin{theorem}\label{th:thne4}
Let $\rho_{\vec{\lambda}}$ and $\sigma_{\vec{\lambda}}$ be given respectively by Eqs.\ \eqref{mems} and~\eqref{eq:sigma}. Then, there exists no NE map $\Lambda:\mathcal{D}\to\mathcal{D}$ such that $\Lambda(\rho_{\vec{\lambda}})=\sigma_{\vec{\lambda}}$ in any of the spectral regions given by
\begin{align}
A&=\{2\lambda_2+\lambda_3-\lambda_1<0, 2\lambda_3+\lambda_4-\lambda_2<0\},\nonumber\\
B&=\{2\lambda_2+\lambda_3-\lambda_1<0, 2\lambda_3+\lambda_4-\lambda_2\geq 0,\lambda_3\leq 2\lambda_4,\nonumber\\
&\lambda_2>\lambda_3+\lambda_4\},\nonumber\\
C&=\{2\lambda_2+\lambda_3-\lambda_1<0 ,2\lambda_3+\lambda_4-\lambda_2\geq 0,\lambda_3 > 2\lambda_4,\nonumber\\
&\lambda_2>\frac{3}{2}\lambda_3\}.
\end{align}
\end{theorem}

\begin{proof}
Equation \eqref{1} for $i=1$ reads
\begin{equation}\label{th11:eq1}
\lambda_1 = \lambda_1 F_{11}+\lambda_2F_{12}+\lambda_3 F_{13}+\lambda_4 F_{14}.
\end{equation}
Using Eq.~\eqref{4} for $i=1$ this then implies that
\begin{equation}
\lambda_1 \leq  \lambda_3 +(\lambda_1-\lambda_3)F_{11}+\lambda_2 F_{12}+ \lambda_4 F_{14},
\end{equation}
which can then be rewritten as
\begin{align}
\lambda_1-\lambda_3 &\leq  (\lambda_1-\lambda_3)\left(F_{11} +\frac{F_{12}+F_{14}}{2}\right)\nonumber\\
&+\left(\lambda_2 -\frac{\lambda_1-\lambda_3}{2}\right)F_{12}+ \left(\lambda_4 -\frac{\lambda_1-\lambda_3}{2}\right) F_{14}.
\end{align}
Equation \eqref{6} for $i=1$ then implies that
\begin{equation}\label{eq:eq1_th11}
0 \leq  \left(\frac{2\lambda_2+\lambda_3-\lambda_1}{2}\right)F_{12}+ \left(\frac{2\lambda_4+\lambda_3-\lambda_1}{2}\right)F_{14}.
\end{equation} 
Observe now that
\begin{eqnarray}
2\lambda_4+\lambda_3-\lambda_1 &=&\lambda_4+(1-\lambda_2-2\lambda_1)\nonumber\\
&=&\lambda_4-\lambda_2+(1-2\lambda_1),
\end{eqnarray}
is always negative for $\lambda_1> \frac{1}{2}$. Thus, if in addition it holds that $2\lambda_2+\lambda_3-\lambda_1<0$, we must necessarily have
that $F_{12}=F_{14}=0$, and, then, by Eqs.\ (\ref{th11:eq1})  and (\ref{4}), that $F_{13}=0$ and $F_{11}=1$. This, in turn, leads to $F_{21}=F_{31}=F_{41}=0$. With these constrains, Eq.~\eqref{1} for $i=2$ reads
\begin{eqnarray}\label{th11:eq2}
\lambda_2=\lambda_2 F_{22}+\lambda_3 F_{23}+ \lambda_4 F_{24},
\end{eqnarray}
which gives
\begin{eqnarray}
F_{23}\geq \frac{\lambda_2-\lambda_4}{2\lambda_3}.
\end{eqnarray}
For this to have a valid solution (i.e., fulfilling $F_{23}\leq1$), we must then have
\begin{eqnarray}
2\lambda_3+\lambda_4-\lambda_2\geq 0.
\end{eqnarray}
This proves that there cannot exist a NE map implementing the desired transformation for spectra in the region $A$.

We assume now that the first two inequalities defining the region $B$ hold and use Eq.~\eqref{7} for $i=2$ to obtain from Eq.\ (\ref{th11:eq2}) that
\begin{eqnarray}
\lambda_2-\lambda_3 \leq  \left(\lambda_2-\frac{\lambda_3}{2}\right)F_{22}-\left(\frac{\lambda_3}{2}-\lambda_4\right)F_{24}%\nonumber\\
%\frac{\lambda_2}{2}-\frac{3}{4}\lambda_3 &\leq & -\left(\frac{\lambda_3}{2}-\lambda_4\right)F_{24}.\nonumber\\
\end{eqnarray}
must hold. Now, imposing that Eq.\ (\ref{5}) should be satisfied entails that
\begin{eqnarray}\label{th11:eq3}
2\lambda_2-3\lambda_3+(2\lambda_3-4\lambda_4)F_{24} &\leq & 0.
\end{eqnarray}
If we assume that $\lambda_3-2\lambda_4\leq 0$ is true, using that $F_{24}\leq \frac{1}{2}$ (Eq.~\eqref{5}) we then have that
\begin{eqnarray}
\lambda_2-\lambda_3-\lambda_4\leq 0
\end{eqnarray}
must hold. This proves the claim in the spectral region given by $B$.

Finally, if we on the other hand assume that $\lambda_3-2\lambda_4>0$ is true in Eq.\ (\ref{th11:eq3}), using that $F_{24}\geq0$ (Eq.~\eqref{3}) we then have that
\begin{eqnarray}
 2\lambda_2-3\lambda_3 &\leq & 0.
\end{eqnarray}
Thus, there cannot be a solution in the region $C$ either and this concludes the proof.
\end{proof}

\section{Conclusions}

Whereas Ref.~\cite{nomems} proved the impossibility of having a maximally entangled mixed state for a fixed spectrum in general by considering some particular eigenvalue distributions for rank-2 two-qubit states, in this work we have investigated whether this notion is possible at all beyond the case of pure states (i.e., rank equal to one). In the two-qubit case, this question boils down to the possibility of transforming the so-called MEMS of Eq.\ (\ref{mems}) to any other isospectral state employing LOCC protocols. We have generalized the technique used in \cite{nomems} based on NE convertibility to define a LP, whose infeasibility implies the impossibility of transforming the MEMS into a particular isospectral target state. We have also provided new techniques that make it possible to conclude that certain transformations from the MEMS are impossible by SEP maps. Since SEP and NE operations are relaxations of LOCC, whereby a negative answer precludes LOCC convertibility, this has allowed us to prove that a maximally entangled two-qubit state for a given spectrum cannot exist in all cases where the rank is two or three as well as for a considerable class of eigenvalue distributions in the case where the rank is four. While the presence of blue and green patches displayed in Fig.~\ref{fig:fig1} still leaves the question open in these spectral regions in the latter case, our results give a clear perspective on the problem at hand. They show that not only a notion of maximal entanglement for a fixed spectrum does not always exist, but that it never does in most cases. These findings indicate that even if it turns out that some eigenvalue distributions happen to permit such a notion beyond the case of pure states, they can only have a marginal impact.

Our results manifest that transformations under LOCC between mixed states share little analogy with the case of pure states and preclude generalizations of Nielsen's theorem even in restricted settings such as eigenvalue-preserving conversions. In addition to this, the techniques we have put forward might find application in the general problem of concluding that certain state transformations among mixed states cannot be implemented by LOCC. In particular, notice that the impossibility results of Sec.\ III do not require the isospectrality condition and that they even apply to stochastic LOCC (i.e., probabilistic LOCC transformations with non-zero probability of success) by considering instead trace-non-increasing maps. Finally, as discussed in \cite{nomems}, whenever a maximally entangled two-qubit state for a fixed spectrum does not exist, then there must exist an entanglement measure that has a different maximizer than the state of Eq.\ (\ref{mems}). However, it would be interesting to find such a measure with a clear operational meaning, i.e., a particular task of practical relevance where there is an isospectral state that outperforms the so-called MEMS.

\begin{acknowledgments}This project was made possible by the DLR Quantum Computing Initiative and the Federal Ministry for Economic Affairs and Climate Action; \url{qci.dlr.de/projects/IQDA}. J.I. de V. acknowledges financial support from the Spanish Ministerio de Ciencia e Innovaci\'on (grants PID2023-146758NB-I00 and PID2024-160539NB-I00, and ``Severo Ochoa Programme for Centres of Excellence" grant CEX2023-001347-S funded by MCIN/AEI/10.13039/501100011033) and from Comunidad de Madrid (grant QUITEMAD-CM TEC-2024/COM-84).
\end{acknowledgments}

\end{document}